\documentclass[a4paper]{article}

\usepackage{graphicx}
\usepackage[margin=1in]{geometry}
 
\usepackage{microtype} %if unwanted, comment out; but I think it's useful (BA)

\usepackage{cite}
\usepackage[ruled,noend,linesnumbered]{algorithm2e}
\usepackage{xspace}
\usepackage{algorithmic}
\usepackage{graphicx}
\usepackage{amsmath,amsthm}

\graphicspath{ {./img/} }

\newtheorem{theorem}{Theorem}
\newtheorem{lemma}[theorem]{Lemma}
\newtheorem{corollary}[theorem]{Corollary}
\newtheorem{observation}[theorem]{Observation}

\newcommand{\med}{\ensuremath{\lfloor n/2 \rfloor}}
\newcommand{\medup}{\ensuremath{\lceil n/2 \rceil}}

\newcommand{\vmed}{\ensuremath{v_{\lfloor n/2 \rfloor}}}

\newcommand{\etal}{\emph{et~al.}\xspace}

\DeclareMathOperator{\SPT}{SPT}

\title{Time-Space Trade-off Algorithms for Triangulating a Simple Polygon\footnote{Work on this paper by B.~A.~was supported in part by NSF Grants CCF-11-17336, CCF-12-18791, and CCF-15-40656, and by grant 2014/170 from the US-Israel Binational Science Foundation.
M.~K.~was partially supported by MEXT KAKENHI grant Nos.~12H00855, and 17K12635.
S.~P.~was supported in part by the Ontario Graduate Scholarship and The Natural Sciences and Engineering Research Council of Canada. A.~v.~R. and M.~R. were supported by JST ERATO Grant Number JPMJER1305, Japan.
An earlier version of this work appeared in the \emph{Proceedings of the 15th Scandinavian Symposium and Workshops on Algorithm Theory}~\cite{akprr-tstotsp-16}.}}

\author{%
  Boris Aronov\thanks{%
    Department of Computer Science and Engineering,
    Tandon School of Engineering, New York University,
    Brooklyn, 11201 USA; \texttt{boris.aronov@nyu.edu}.}
  \and
  Matias Korman\thanks{Tohoku University, Sendai, Japan;
    \texttt{mati@dais.is.tohoku.ac.jp}.}
  \and
  Simon Pratt\thanks{Cheriton School of Computer Science, University of Waterloo, Canada; \texttt{Simon.Pratt@uwaterloo.ca}.}
  \and
  Andr\'e~van~Renssen\thanks{%
    National Institute of Informatics (NII), Tokyo, Japan;
    \texttt{\{andre,marcel\}@nii.ac.jp}.}\, $^{,}$\thanks{JST, ERATO, Kawarabayashi Large Graph Project.}
  \and
  Marcel Roeloffzen\footnotemark[5]\, $^{,}$\footnotemark[6]
}

\date{}

\begin{document}

\maketitle

\begin{abstract}
  An $s$-workspace algorithm is an algorithm that has read-only access to the values of the input, write-only access to the output, and only uses $O(s)$ additional words of space. We present a randomized $s$-workspace algorithm for triangulating a simple polygon~$P$ of $n$~vertices that runs in $O(n^2/s+n \log n \log^{5} (n/s))$ expected time using $O(s)$ variables, for any $s \leq n$. In particular, when
  $s \leq \frac{n}{\log n\log^{5}\log n}$
  the algorithm runs in $O(n^2/s)$ expected time.
\end{abstract}

\section{Introduction}

Triangulation of a simple polygon, often used as a preprocessing step in computer graphics, is performed in a wide range of settings including on embedded systems like the Raspberry Pi or mobile phones.  Such systems often run read-only file systems for security reasons and have very limited working memory.  An ideal triangulation algorithm for such an environment would allow for a trade-off in performance in time versus working space.

Computer science and specifically the field of algorithms generally have two optimization goals; running time and memory size. In the 70's there was a strong focus on algorithms that required low memory as it was expensive. As memory became cheaper and more widely available this focus shifted towards optimizing algorithms for their running time, with memory mainly as a secondary constraint. 

Nowadays, even though memory is cheap, there are other constraints that limit memory usage. First, there is a vast number of embedded devices that operate on batteries and have to remain small, which means they simply cannot contain a large memory. Second, some data may be read-only, due to hardware constraints (e.g., read-only or write-once DVDs/CDs) or concurrency issues (i.e., to allow many processes to access the database at once). 

These memory constraints can all be described in a simple way by the so-called \emph{constrained-workspace} model (see Section~\ref{sec:prelim} for details). Our input is read-only and potentially much larger than our working space, and the output we produce must be written to write-only memory. More precisely, we assume we have a read-only data set of size $n$ and a working space of size $O(s)$, for some user-specified parameter $s$. In this model, the aim is to design an algorithm whose running time decreases as $s$ grows. Such algorithms are called \emph{time-space trade-off} algorithms~\cite{s-mcepc-08}.

\subsection*{Previous Work}
Several models of computation that consider space constraints have been studied in the past (we refer the interested reader to~\cite{k-mca-15} for an overview). In the following we discuss the results related to triangulations. The concept of memory-constrained algorithms attracted renewed attention  within the computational geometry community by the work of Asano~\etal~\cite{amrw-cwagp-10}. One of the algorithms presented in~\cite{amrw-cwagp-10} was for triangulating a set of $n$ points in the plane in $O(n^2)$ time using $O(1)$ variables. More recently, Korman~\etal~\cite{kmrrss-tstotvd-15} introduced two different time-space trade-off algorithms for triangulating a point set: the first one computes an arbitrary triangulation in $O(n^2/s + n\log^2 n)$ time using $O(s)$ variables. The second is a randomized algorithm that computes the Delaunay triangulation of the given point set in expected $O((n^2/s)\log s + n\log n\log^*n)$ time within the same space bounds.

The above results address triangulating discrete point sets in the plane. The first algorithm in this model for triangulating simple polygons was due to Asano~\etal~\cite{abbkmrs-mcasp-11} (in fact, the algorithm works for slightly more general inputs: plane straight-line graphs). It runs in $O(n^2)$ time using $O(1)$ variables. The first time-space trade-off for triangulating polygons was provided by Barba~\etal~\cite{bklss-sttosba-14}. In their work, they describe a general time-space trade-off algorithm that in particular could be used to triangulate monotone polygons. An even faster algorithm (still for monotone polygons) was afterwards found by Asano and Kirkpatrick~\cite{ak-tstanlnp-13}: $O(n\log_sn)$ time using $O(s)$ variables. Despite extensive research on the problem, there was no known time-space trade-off algorithm for general simple polygons. It is worth noting that no lower bounds on the time-space trade-off are known for this problem either.

If we forego space constraints, we can triangulate a simple polygon of $n$ vertices in linear time (using linear space)~\cite{c-tsplt-91}. However, this algorithm is considered difficult to implement and very slow in practice (see, e.g.,~\cite[p.~57]{o_rourke_c}). Alternatively, Hertel and Mehlhorn~\cite{hertel_mehlhorn} provided an algorithm that can triangulate a simple polygon of $n$ vertices, $r$ of which are reflex, in~$O(n \log r)$ time. Since our work is of theoretical nature, we will use Chazelle's triangulation algorithm. However, as the running time of our algorithms is dominated by other terms, we can instead use the one of Hertel and Mehlhorn without affecting the asymptotic performance.

\subsection*{Results}

This paper is structured as follows. In Section~\ref{sec:prelim} we define our model, as well as the problems we study. Our main result on triangulating a simple polygon $P$ with $n$ vertices using only a limited amount of memory can be found in Section~\ref{sec:main}. Our algorithm achieves expected running time of $O(n^2/s+n \log n \log^{5} (n/s))$ using $O(s)$ variables, for any $s \leq n$. Note that for most values of $s$ (i.e., when $s \leq \frac{n}{\log n \log^{5}\log n}$) the algorithm runs in $O(n^2/s)$ expected time.

Our approach uses a recent result by Har-Peled~\cite{Har-Peled15} as a tool for subdividing $P$ into smaller pieces and solving them recursively. Once the pieces are small enough to fit into memory, the subproblem can be handed over to the usual algorithm, without memory constraints. This divide-and-conquer approach has been often used in the memory-constrained literature, but each time the partition was constructed \emph{ad~hoc}, based on the properties of the problem being solved.  We believe that the tool we introduce in this paper is very general and can be used for several problems. As an example, in Section~\ref{sec:extensions} we show how the same approach can be used to compute the \emph{shortest-path tree} from any point $p\in P$, or simply to split $P$ by $\Theta(s)$ pairwise non-crossing diagonals into smaller subpolygons, each with $\Theta(n/s)$ vertices. 

\section{Preliminaries}\label{sec:prelim}

In this paper, we utilize the $s$-\emph{workspace} model of computation that is frequently used in the literature (see, for example, \cite{abbkmrs-mcasp-11,bklss-sttosba-14,bkls-cvpufv-13,Har-Peled15}). In this model the input data is given in a read-only array or some similar structure. In our case, the input is a simple polygon $P$; let~$v_1, v_2, \ldots, v_n$ be the vertices of $P$ in clockwise order along its boundary. We assume that, given an index~$i$, in constant time we can access the coordinates of the vertex $v_i$. We also assume that the usual word RAM operations (say, given $i$, $j$, $k$, finding the intersection point of the line passing through vertices $v_i$ and $v_j$ and the horizontal line passing through $v_k$) can be performed in constant time.

In addition to the read-only data, an $s$-workspace algorithm can use $O(s)$ variables during its execution, for some parameter $s$ determined by the user. Implicit memory consumption (such as the stack space needed in recursive algorithms) must be taken into account when determining the size of a workspace. We assume that each variable or pointer is stored in a data word of $\Theta(\log n)$ bits. Thus, equivalently, we can say that an $s$-workspace algorithm uses $O(s\log n)$ bits of storage.

In this model we study the problem of computing a \emph{triangulation} of a simple polygon~$P$, which is a maximal crossing-free straight-line graph whose vertices are the vertices of~$P$ and whose edges lie inside~$P$. Unless $s$ is very large, the triangulation cannot be stored explicitly. Thus, the goal is to report a triangulation of $P$ in a write-only data structure. Once an output value is reported, it cannot be accessed or modified afterwards. 

In other memory-constrained triangulation algorithms~\cite{abbkmrs-mcasp-11,ak-tstanlnp-13} the output is reported as a list of edges in no particular order, with no information on neighboring edges or faces. Moreover, it is not clear how to modify these algorithms to obtain such information. Our approach has the advantage that, in addition to the list of edges, we can also report the triangles generated, together with the adjacency relationship between the edges and the triangles; see Section~\ref{sec:output} for details.

A vertex of a polygon is \emph{reflex} if its interior angle is larger than $180^\circ$.
Given two points $p,q\in P$, the \emph{geodesic} (or \emph{shortest path}) between them is the path of minimum length that connects $p$ and $q$ and that stays within $P$ (viewing $P$ as a closed set). The length of that path is the \emph{geodesic distance} from $p$ to $q$. It is well known that, for any two points of~$P$, their geodesic $\pi$ always exists and is unique. Such a path is a polygonal chain whose vertices (other than $p$ and $q$) are reflex vertices of $P$. Thus, we often identify $\pi$ with the ordered sequence of reflex vertices traversed by the path from $p$ to $q$. When that sequence is empty (i.e., the geodesic consists of the straight segment $pq$) we say that $p$ \emph{sees} $q$ and vice versa.

Our algorithm relies on a recent procedure by Har-Peled~\cite{Har-Peled15} for computing geodesics under memory constraints, which constructs the geodesic between any two points in a simple polygon of $n$ vertices in expected $O(n^2/s+n\log s \log^4 (n/s))$ time using $O(s)$ words of space. Note that this path might not fit in memory, so the edges of the geodesic are reported one by one, in order. 

\section{Algorithm}\label{sec:main}

Let $\pi$ be the
geodesic connecting $v_1$ and $\vmed$. From a high-level perspective,
the algorithm uses the approach of Har-Peled~\cite{Har-Peled15} to
compute $\pi$. We will use the computed edges to subdivide $P$ into smaller problems that can be solved recursively.

We start by introducing some definitions that will help in recording
the portion of the polygon already triangulated. 
Vertices $v_1$ and
$\vmed$ split the boundary of $P$ into two chains. We say $v_i$ is a
\emph{top} vertex if $1 < i < \med$ and a \emph{bottom} vertex if
$\med < i \leq n$. Top/bottom is the \emph{type} of a vertex and all vertices (except for $v_1$ and $\vmed$) have exactly one type. A diagonal $c$ is \emph{alternating} if it
connects a top and a bottom vertex or if one of its endpoints is either $v_1$ or $\vmed$, and \emph{non-alternating} otherwise.

We will use diagonals to partition $P$ into two parts. For simplicity of the exposition, given a diagonal $d$, we regard both components of $P\setminus d$ as closed (i.e., the diagonal belongs to both of them).  Since any two consecutive vertices of $P$ can see each other, the partition produced by an edge of $P$ is trivial, in the sense that one subpolygon is $P$ and the other one is a line segment. 

\begin{observation}\label{obs_split}
Let $c$ be a diagonal of $P$ not incident to $v_1$ or
$\vmed$. Vertices $v_1$ and $\vmed$ belong to different components of
$P \setminus c$ if and only if $c$ is an alternating diagonal.
\end{observation}

\begin{corollary}\label{cor_nonalter}
Let $c$ be a non-alternating diagonal of $P$. The component of
$P \setminus c$ that contains neither $v_1$ nor $\vmed$ has at most
$\medup$ vertices.
\end{corollary}

We will use alternating diagonals as a way to remember what part of the polygon has already been triangulated.  More specifically, the algorithm will at all times store an alternating diagonal $a_c$. An invariant of our algorithm is that the connected component of $P\setminus a_c$ not containing $\vmed$ has already been triangulated. 

Ideally, $a_c$ would be an edge of $\pi$, the geodesic connecting $v_1$ and $\vmed$, but this is not always possible. Instead, we guarantee that at least one of the endpoints of $a_c$ is a vertex of $\pi$ that has already been computed in the execution of the shortest-path algorithm.

With these definitions in place, we can give an intuitive description of our algorithm. We start by setting $a_c$ as the degenerate diagonal from $v_1$ to $v_1$. We then use the shortest-path computation procedure of Har-Peled. Our aim is to walk along $\pi$ until we find a new alternating diagonal $a_{\textrm{new}}$. At that moment we pause the execution of the shortest-path algorithm, triangulate the subpolygons of $P$ that have been created (and contain neither $v_1$ nor $\vmed$) recursively, set $a_c$ to $a_{\textrm{new}}$, and resume the execution of the shortest-path algorithm. 

Although our approach is intuitively simple, there are several technical difficulties that must be carefully considered. Ideally, the number of vertices we walk along $\pi$ before finding an alternating diagonal is small and thus they can be stored explicitly.  But if we do not find an alternating diagonal on $\pi$ in just a few steps (indeed, $\pi$ may contain no alternating diagonal), we need to use other diagonals. We also need to make sure that the complexity of each recursive subproblem is reduced by a constant fraction, that we never exceed space bounds, and that no part of the triangulation is reported more than once.

Let $v_c$ denote the endpoint of $a_c$ that is on $\pi$ and that is closest to $\vmed$. Recall that
the subpolygon defined by $a_c$ containing $v_1$ has already been triangulated. Let $w_0, \ldots , w_k$ be the portion of $\pi$ up to the next alternating diagonal. That is, path $\pi$ is of the form $\pi=(v_1, \ldots, v_c=w_0, w_1, \ldots, w_{k}, \ldots, \vmed)$ where $w_1, \ldots, w_{k-1}$ are of the same type as $v_c$, and $w_k$ is of different type (or $w_k=\vmed$ if all vertices between $v_c$ and $\vmed$ are of the same type). 

Consider the partition of $P$ induced by $a_c$ and this portion of $\pi$; see Figure~\ref{fig:ChainBetweenAlternationDiagonals2}. Let $P_1$ be the subpolygon induced by $a_c$ that does not contain $\vmed$. Similarly, let $P_{\med}$ be the subpolygon that is induced by the alternating diagonal $w_{k-1}w_k$ and does not contain $v_1$.\footnote{For simplicity of the exposition, the definition of $P_1$ assumes that $\vmed$ is not an endpoint of $a_c$ (similarly, $v_1$ not an endpoint of $w_{k-1}w_k$ in the definition of $P_{\med}$). Each of these conditions is not satisfied once (i.e., at the first and last diagonals of $\pi$), and in those cases the polygons $P_1$ and $P_{\med}$ are not properly defined. Whenever this happens we have $k=1$ and a single diagonal that splits $P$ in two. Thus, if $\vmed\in a_c$ (and thus $P_1$ is undefined), we simply define $P_1$ as the complement $P_{\med}$ (similarly, if $v_1\in w_{k-1}w_k$, we define $P_{\med}$ as complement of $P_1$). If both subpolygons are undefined simultaneously we assign them arbitrarily.} For any $i<k-1$, we define $Q_i$ as the subpolygon induced by the non-alternating diagonal~$w_iw_{i+1}$ that contains neither $v_1$ nor $\vmed$. Finally, let $R$ be the remaining component of $P$.  Some of these subpolygons may be degenerate and consist only of a line segment (for example, when $w_iw_{i+1}$ is an edge of $P$).

\begin{figure}
  \centering
  \includegraphics{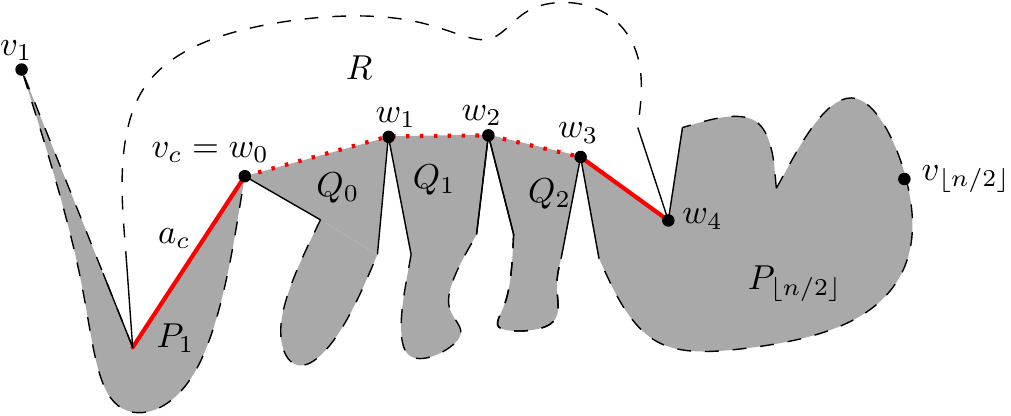}
  \caption{Partitioning $P$ into subpolygons $P_1$, $P_{\med}$, $R$, $Q_1$, $\ldots$, $Q_{k-2}$. The two alternating diagonals are marked by thick red lines.}
  \label{fig:ChainBetweenAlternationDiagonals2}
\end{figure}

\begin{lemma}
\label{lem:closeAlternation}
Each of the subpolygons $R$, $Q_1$, $Q_2$, $\ldots$, $Q_{k-2}$ has at most $\medup+k$ vertices. Moreover, if $w_k=\vmed$, then the subpolygon $P_{\med}$ has at most $\medup$ vertices. 
\end{lemma}

\begin{proof}
Subpolygons $Q_i$ are induced by non-alternating diagonals and cannot have more than $\medup$ vertices, by Corollary~\ref{cor_nonalter}. The proof for $R$ follows by definition: the boundary of $R$ comprises the path $w_0 \dots w_k$ and a contiguous portion of $P$ consisting of only top vertices or only bottom vertices. Recall that there are at most $\medup$ of each type. Similarly, if $w_k=\vmed$, subpolygon $P_{\med}$ can only have vertices of one type (either only top or only bottom vertices), and thus the bound holds. This completes the proof of the Lemma.
\end{proof}

This result allows us to treat the easy case of our algorithm. When $k$ is small (say, a constant), we can pause the shortest-path computation, explicitly store all vertices $w_i$, recursively triangulate $R$ as well as the subpolygons $Q_i$ (for all $i\leq k-2$), update $a_c$ to the edge $w_{k-1}w_k$, and resume the shortest-path algorithm. 

Handling the case of large $k$ is more involved. Note that we do not know the value of~$k$ until we find the next alternating diagonal, but we need not compute it directly.  Given a parameter $\tau$ related to the workspace allowed for our algorithm, we say that the path is \emph{long} when $k>\tau$. Initially we set $\tau=s$ but the value of this parameter will change as we descend the recursion tree.  We say that the distance between two alternating diagonals is \emph{long} whenever we have computed $\tau$ vertices of $\pi$ beyond $v_c$ and they are all of the same type as $v_c$. That is, path $\pi$ is of the form $\pi=(v_1, \ldots, v_c=w_0, w_1, \ldots, w_{\tau}, \ldots \vmed)$ and vertices $w_0, w_1, \ldots w_{\tau}$ have the same type and, in particular, form a convex chain (see Figure~\ref{fig:ChainBetweenAlternationDiagonals2}). Rather than continue walking along $\pi$, we look for a vertex~$u$ of~$P$ that together with $w_{\tau}$ forms an alternating diagonal. Once we have found this diagonal, we have at most $\tau+2$ diagonals ($a_c, w_0w_1, w_1w_2, \ldots, w_{\tau-1}w_{\tau}$, and $uw_{\tau}$) partitioning $P$ into at most $\tau+3$ subpolygons once again: $P_1$ is the part induced by $a_c$ which does not contain~$\vmed$, $P_{\med}$ is the part induced by $uw_{\tau}$ which does not contain~$v_1$, $Q_i$ is the part induced by $w_iw_{i+1}$, which contains neither $v_1$ nor $\vmed$, and $R$ is the remaining component.
 
\begin{lemma}
\label{lem:farAlternation}
We can find a vertex $u$ so that $uw_{\tau}$ is an alternating diagonal, in $O(n)$ time using $O(1)$ space. Moreover, each of the subpolygons $R$, $Q_1$, $Q_2$, $\ldots$, $Q_{\tau-2}$ has at most $\medup+\tau$ vertices.
\end{lemma}

\begin{figure}
  \centering
  \includegraphics[width=0.45\textwidth]{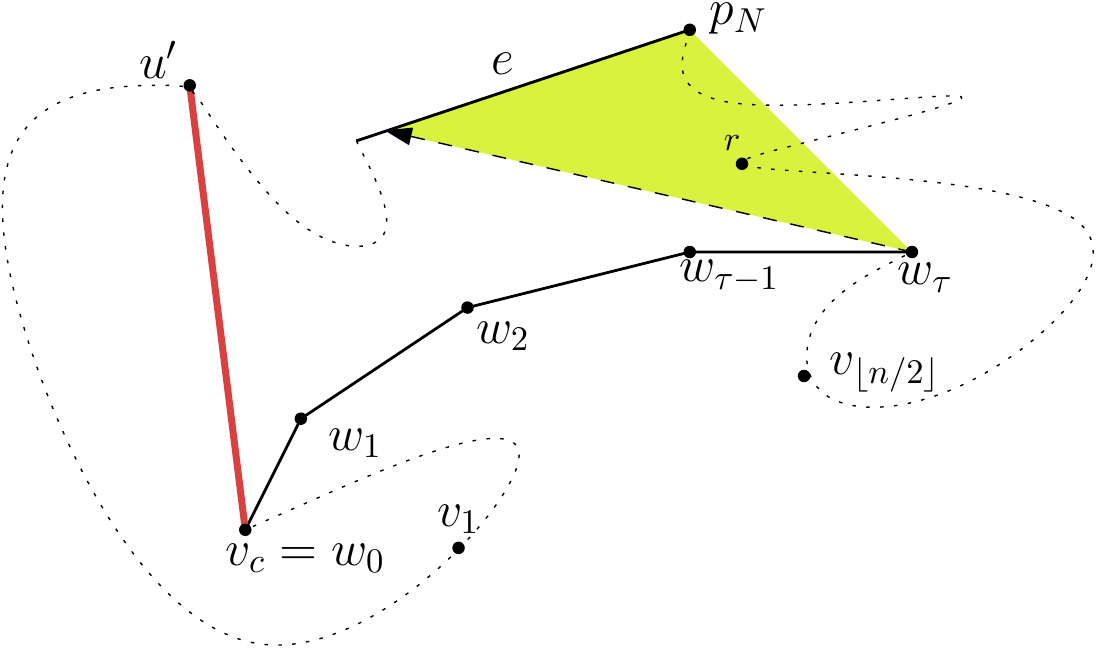}\qquad%
  \includegraphics[width=0.45\textwidth]{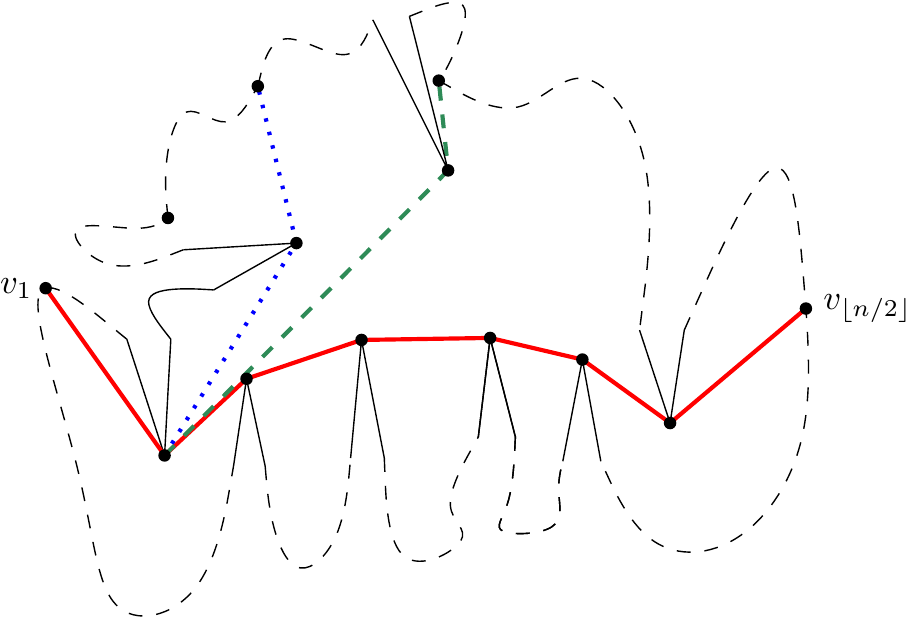}
  \caption{(left)~After we have walked $\tau$ steps of $\pi$, we can find an alternating diagonal by shooting a ray from $w_{\tau}$ either towards $u'$ or $w_{\tau-1}$ (whichever is higher). If $p_N$ is not visible (as it happens in this example), there is a reflex vertex $r$ within the triangular region that is visible from $w_r$. The upper endpoint $p_N$ of the first edge $e$ hit might not be visible. But then the reflex vertex of largest angle inside the triangular zone must be visible to $w_\tau$.
  (right)~At a different level of recursion the subproblems are formed by a consecutive chain of the input and a list of $O(\tau)$ cut vertices. The geodesics used to split the problem at first, second, and third level are depicted in solid red, dashed green, and dotted blue, respectively.}
  \label{fig:rayshooting}
\end{figure}

\begin{proof}
Proofs for the size of the subpolygons are identical to those of Lemma~\ref{lem:closeAlternation}. Thus, we focus on how to compute $u$ efficiently. 
Without loss of generality, we may assume that the edge $w_{\tau-1}w_{\tau}$ is horizontal. Recall that the chain $w_0, \ldots, w_{\tau}$ is in convex position, thus all of these vertices must lie on one side of the line $\ell$  through $w_{\tau-1}$ and $w_{\tau}$, say below it.
Let $u'$ be the endpoint of $a_c$ other than $v_c$. If $u'$ also lies below $\ell$, we shoot a ray from $w_{\tau}$ towards~$w_{\tau-1}$. Otherwise, we shoot a ray from $w_{\tau}$ towards $u'$. Let $e$ be the first edge that is properly intersected by the ray and let $p_N$ be the endpoint of $e$ of highest $y$-coordinate. Observe that $p_N$ must be on or above $\ell$; see Figure~\ref{fig:rayshooting}~(left).

Ideally, we would like to report $p_N$ as the vertex $u$. However, point $p_N$ need not be visible even though some portion of $e$ is. Whenever this happens, we can use the visibility properties of simple polygons: since $e$ is partially visible, the portion of $P$ that obstructs visibility between $w_{\tau}$ and $p_N$ must cross the segment from $w_{\tau}$ to $p_N$. In particular, there must be one or more reflex vertices in the triangle formed by $w_{\tau}$, $p_N$, and the visible point of $e$ (shaded region of Figure~\ref{fig:rayshooting}~(left)). Among those vertices, the vertex $r$ that maximizes the angle $\angle p_Nw_{\tau}r$ must be visible from $w_\tau$ (see Lemma~1 of \cite{bkls-cvpufv-13}). 

We claim that $r$ must be a top vertex: otherwise $\pi$ would need to pass through $r$ to reach $\vmed$, but since $r$ is above $\ell$, the shortest path from $w_{\tau-1}$ to $r$ does not go through~$w_\tau$.  This means that $\pi$ could be made shorter by taking the shortest path from $w_{\tau-1}$ to $r$ instead of going through $w_\tau$.
This contradicts $\pi$ being the shortest path between $v_1$ and $\vmed$, and thus we conclude that $r$ is a top vertex, as claimed. 

As described in Lemma~1 of \cite{bkls-cvpufv-13}, in order to find such a reflex vertex we need to scan~$P$ at most three times, each time storing a constant amount of information: once for finding the edge $e$ and point $p_N$, once more to determine if $p_N$ is visible, and a third time to find $r$ if $p_N$ is not visible.
\end{proof}

At a high level, our algorithm walks from $v_1$ to $\vmed$, stopping after walking $\tau$ steps or after finding an alternating diagonal, whichever comes first. This generates several subproblems of smaller complexity that are solved recursively. Once the recursion is done we update $a_c$ (to keep track of the portion of $P$ that has been triangulated), and continue walking along $\pi$. The walking process ends when it reaches $\vmed$. In this case, in addition to triangulating $R$ and the  subpolygons $Q_i$ as usual, we must also triangulate $P_{\med}$. 

The algorithm at the deeper levels of recursion is almost identical. There are only three minor changes that need to be introduced. We need some base cases to end the recursion. Recall that $\tau$ denotes the amount of space available to the current instance of the problem. Thus, if $\tau$ is comparable to $n$ (say, $10\tau \geq n$), then the whole polygon fits into memory and can be triangulated in linear time~\cite{c-tsplt-91}. Similarly, if $\tau$ is small (say, $\tau \leq 10$), we have run out of space and thus we triangulate $P$ using a constant-workspace algorithm~\cite{abbkmrs-mcasp-11}. In all other cases we continue with the recursive algorithm as usual. 

For ease in handling the subproblems, at each step we also designate the vertex that fulfills the role of $v_1$ (i.e., one of the vertices from which the geodesic must be computed). Recall that we have random access to the vertices of the input. Thus, once we know which vertex plays the role of $v_1$, we can find the vertex that plays the role of $\vmed$ in constant time as well.

\begin{algorithm}%[tb]
  \renewcommand{\algorithmiccomment}[1]{(* #1 *)}
  \begin{algorithmic}[1]  
	\IF[The polygon fits into memory.]{$10\tau \geq n$}
			\STATE  Triangulate $P$ using Chazelle's algorithm~\cite{c-tsplt-91}
	\ELSIF[We ran out of recursion space.]{$\tau \leq 10$}
			\STATE Triangulate $P$ using the constant workspace algorithm~\cite{abbkmrs-mcasp-11}
	\ELSE[$P$ is large, we will use recursion.]
	\STATE $a_c \leftarrow v_1v_1$
	\STATE $v_c \leftarrow v_1$	
	\STATE walked $\leftarrow v_1$ \COMMENT{Variable to keep track of how far we have walked on $\pi$.}
	\WHILE{walked $\neq \vmed$} 
		\STATE $i \leftarrow 0$ (* $i$ counts the number of steps before finding an alternating diagonal *)
		\REPEAT 
			\STATE $i \leftarrow i+1$
			\STATE $w_i \leftarrow $ next vertex of $\pi$ 
		\UNTIL{$i=\tau$ \OR $w_{i-1}w_i$ is an alternating diagonal}
		\IF{$w_{i-1}w_i$ is an alternating diagonal}
			\STATE $u' \leftarrow w_{i-1}$
			\STATE $a_{\textrm{new}} \leftarrow w_iw_{i-1}$
		\ELSE[We walked too much. Use Lemma~\ref{lem:farAlternation} to partition the problem.]
			\STATE $u' \leftarrow$ \textsc{FindAlternatingDiagonal}$(P, a_c, v_c, w_1, \ldots, w_{\tau})$
			\STATE $a_{\textrm{new}} \leftarrow u'w_i$
		\ENDIF
		\STATE \COMMENT{Now we triangulate the subpolygons.}
		\STATE Triangulate$(R,u',\tau \cdot \kappa)$ 
		\FOR{$j=0$ \TO $i-2$}
		\STATE Triangulate$(Q_j,w_j,\tau \cdot \kappa)$
		\ENDFOR
		\STATE $a_c \leftarrow a_{\textrm{new}}$
		\STATE $v_c \leftarrow w_{\tau}$	
		\STATE walked $\leftarrow w_{\tau}$		
	\ENDWHILE
	\STATE \COMMENT{We reached $\vmed$.  All parts except $P_{\med}$ have been triangulated.}
	\STATE  Triangulate$(P_{\med},w_i,\tau\cdot \kappa)$
	\ENDIF

  \end{algorithmic}
\caption{Pseudocode for Triangulate$(P,v_1,\tau)$ that, given a simple polygon~$P$ with $n$~vertices, a vertex $v_1$ of $P$, and workspace capacity $\tau$, computes a triangulation of $P$ in $O(n^2/\tau+n \log n \log^{5} (n/\tau))$ expected time using $O(\tau)$ variables.}
\label{algo_main}
\end{algorithm}

In order to avoid exceeding the space bounds, at each level of the recursion we decrease the value of $\tau$ by a factor of $\kappa<1$. The exact value of the constant $\kappa$ will be determined below. Pseudocode of the recursive algorithm can be found in Algorithm~\ref{algo_main}. Although not explicitly defined in pseudocode, procedure $ \textsc{FindAlternatingDiagonal}$ computes an alternating diagonal as described in Lemma~\ref{lem:farAlternation}.

\begin{theorem}\label{main_theo}
Let $P$ be a simple polygon of $n$ vertices. For any $s \leq n$ we can compute a
triangulation of~$P$ in $O(n^2/s+n \log s \log^{5} (n/s))$ expected time using $O(s)$
variables. In~particular, when $s \leq\frac{n}{\log n\log^{5}\log n}$ the algorithm runs in $O(n^2/s)$ expected time.
\end{theorem}

In the remainder of the section we prove correctness of our algorithm and analyze its time and space requirements.

\subsection{Correctness}

We maintain the invariant that the current diagonal $a_c$ records the
already triangulated portion of the polygon.  Every edge we output is a proper diagonal of $P$ and we recurse on subpolygons created by partitioning by such edges.  Thus, we never report
an edge of the triangulation more than once. Hence, in order to show
correctness of the algorithm, it suffices to prove that the recursion
eventually terminates.

During the execution of the algorithm, we invoke recursion for polygons $Q_i$, $R$, and~$P_{\med}$ (the latter one only when we have reached $\vmed$). By Lemma~\ref{lem:closeAlternation} all of these polygons have size at most $n/2+\tau$. Since we only enter this level of recursion whenever $\tau\leq n/10$ (see lines 1--2 of Algorithm~\ref{algo_main}), overall the size of the problem decreases by a factor of $6/10$,
thereby guaranteeing that the recursion depth is bounded by $O(\log n)$.  Note that there are several conditions for stopping the recursion, but only one of them is needed to guarantee $O(\log n)$ depth.

At each level of recursion we use the shortest-path algorithm of Har-Peled. This algorithm needs random access in constant time to the vertices of the polygon. Thus, we must make sure that this property is preserved at all levels of recursion. A simple way to do so would be to explicitly store the polygon in memory at every recursive call, but this may exceed the space bounds of the algorithm.

Instead, we make sure that the subpolygon is described by $O(\tau)$ words. By construction, each subpolygon consists of a single chain of contiguous input vertices of $P$ and at most~$\tau$ additional \emph{cut} vertices (vertices from the geodesics at higher levels).  We can represent the portion of $P$ by the indices of the first and last vertex of the chain and explicitly store the indices of all cut vertices. By an appropriate renaming of the indices within the subpolygon, we can make the vertices of the chain appear first, followed by the cut vertices. Thus, when we need to access the $i$th vertex of the subpolygon, we can check if $i$ corresponds to a vertex of the chain or one of the cut vertices and identify the desired vertex in constant time, in either case.

Now, we must show that each recursive call satisfies this property. Clearly this holds for the top level of recursion, where the input polygon is simply $P$ and no cut vertices are needed. At the next level of recursion each subproblem has up to $\tau$ cut vertices and a chain of contiguous input vertices.  We ensure that this property is satisfied at lower levels of recursion by an appropriate choice of $v_1$ (the vertex from which we start the path): at each level of recursion we build the next geodesic starting from either the first or last cut vertex. This might create additional cut vertices, but their position is immediately after or before the already existing cut vertices (see Figure~\ref{fig:rayshooting}~(right)). 
This way we guarantee constant-time random access to current instance vertices, at all levels of recursion.

\subsection{Time Bounds}\label{sec:time}

We use a two-parameter function $T(\eta,\tau)$ to bound the  expected running time of the algorithm at all levels of recursion. The first parameter $\eta$ represents the size of the problem. Specifically, for a polygon of $n$ vertices we set $\eta=n-2$, namely, the number of triangles to be reported. The second parameter $\tau$ gives the space bound for the algorithm. Initially, we have $\tau=s$, but this value decreases by a factor of $\kappa$ at each level of recursion. Recall that $\tau$ is also the workspace limit for the shortest-path algorithm of Har-Peled that we invoke as part of our algorithm. In addition, $\tau$ is also used as the limit on the length of the geodesic we explore looking for an alternating diagonal. Note that the memory usage of both our algorithm as well as the algorithm by Har-Peled is $O(\tau)$, that is, there are hidden constants. In order to solve the recursions, we cannot use a big-O notation and for readability we assume all hidden constants are 1 and simply write $\tau$ instead.

When $\tau$ becomes really small (say, $\tau \leq 10$) we have run out of allotted space. Thus, we triangulate the polygon using the constant workspace method of Asano \etal~\cite{abbkmrs-mcasp-11} that runs in $O(\eta^2)$ time. Similarly, if the space is large when compared to the instance size (say, $10 \tau\geq \eta$) the polygon fits in the allowed workspace, hence we use Chazelle's algorithm~\cite{c-tsplt-91} for triangulating it. In both cases we have $T(\eta,\tau) \leq c_{\Delta}(\eta^2/\tau+\eta)$ for some constant $c_{\Delta}>0$. 

Otherwise, we partition the problem and solve it recursively. First we bound the time needed to compute the partition. The main tool we use is computing the geodesic between $v_1$ and $\vmed$. This is done by the algorithm of Har-Peled~\cite{Har-Peled15} which takes $O(\eta^2/\tau + \eta \log \tau \log^4 (\eta/\tau))$ expected time and uses $O(\tau)$ space. Recall that we may pause and resume it often during the execution of our algorithm, but overall we only execute it once, not counting recursive calls.

Another operation that we execute is \textsc{FindAlternatingDiagonal} (i.e., Lemma~\ref{lem:farAlternation}) which takes $O(\eta)$ time and $O(1)$ space. In the worst case, this operation is invoked once for every $\tau$ vertices of $\pi$. Since $\pi$ cannot have more than $\eta$ vertices, the overall time spent in this operation is bounded by $O(\eta^2/\tau+\eta)$. Thus, ignoring the time spent in recursion, the expected running time of the algorithm is $c_\textsc{HP}(\eta^2/\tau+ \eta \log \tau \log^4 (\eta/\tau))$ for some constant $c_\textsc{HP}$, which without loss of generality we assume to be at least $c_{\Delta}$.
We thus obtain a recurrence of the form

\[ T(\eta,\tau) \leq c_\textsc{HP}\left(\frac{\eta^2}{\tau}+ \eta \log \tau \log^4 \frac{\eta}{\tau}\right) + \sum_j T(\eta_j, \tau\kappa).
\]
Recall that the values $\eta_j$ cannot be very large, compared to $\eta$. Indeed, each subproblem can have at most a constant fraction $c$ of vertices of the original one (i.e., the way in which lines 1--4 of Algorithm~\ref{algo_main} have been set, we have $c=6/10$). Thus, each $\eta_j$ satisfies $\eta_j \leq c(\eta+2)-2 \leq c\eta$. Since subproblems partition the current polygon, we also have $\sum_j \eta_j = \eta$. 

We claim that there exists a constant $c_R$, so that, for any $\tau,\eta>0$, $T(\eta,\tau) \leq c_R(\eta^2/\tau+ \eta\log \tau \log^{5} (\eta/\tau))$.
Indeed, when $\tau$ is small or the problem size fits into memory (for our choice of constants, this corresponds to $\tau \leq 10$ or $10\tau \geq \eta$) we have $T(\eta,\tau) \leq c_{\Delta}(\eta^2/\tau+\eta) \leq c_R (\eta^2/\tau+\eta)$ for any value of $c_R$ such that $c_R \geq c_{\Delta}$. Otherwise, we use induction and obtain

\begin{align*}
T(\eta,\tau) &\leq c_{\textrm{HP}}\left(\frac{\eta^2}{\tau}+ \eta\log \tau \log^4 \frac{\eta}{\tau}\right) + \sum_j T(\eta_j,\tau\kappa) \\
&\leq c_{\textrm{HP}}\left(\frac{\eta^2}{\tau}+ \eta\log \tau \log^4 \frac{\eta}{\tau}\right) + \frac{c_R}{\tau\kappa}\sum_j \eta_j^2+ c_R \sum_j \eta_j(\log \tau\kappa)\log^{5}\frac{\eta_j}{\tau\kappa} \\
&\leq \left(c_{\textrm{HP}}\frac{\eta^2}{\tau} + \frac{c_R}{\tau\kappa}\sum_j \eta_j^2\right) + c_{\textrm{HP}}\eta\log \tau \log^4\frac{\eta}{\tau} + c_R\sum_j \eta_j\log \tau\log^{5}\frac{\eta_j}{\tau\kappa}\\
&\leq \left(c_{\textrm{HP}}\frac{\eta^2}{\tau} + \frac{c_R}{\tau\kappa}\sum_j \eta_j^2\right) + c_{\textrm{HP}}\eta\log \tau \log^4 \frac{\eta}{\tau} + c_R\sum_j \eta_j\log \tau\log^{5}\frac{c\eta}{\tau\kappa}\\
&\leq \left(c_{\textrm{HP}}\frac{\eta^2}{\tau} + \frac{c_R}{\tau\kappa}\sum_j \eta_j^2\right) + c_{\textrm{HP}}\eta\log \tau \log^4 \frac{\eta}{\tau} +c_R\eta\log \tau\log^{5}\frac{c\eta}{\tau\kappa}.
\end{align*} 
The sum $\sum_j \eta_j^2$ is at most $c\eta \sum_j \eta_j = c\eta^2$, since $\eta_j \leq c\eta$ and $\sum_j \eta_j = \eta$, yielding

\begin{align*}
T(\eta,\tau) &\leq \left(c_{\textrm{HP}}\frac{\eta^2}{\tau} + \frac{c_Rc}{\kappa}\frac{\eta^2}{\tau}\right) + c_{\textrm{HP}}\eta\log \tau \log^4\frac{\eta}{\tau} +c_R\eta\log \tau\log^{5}\frac{c\eta}{\tau\kappa}\\
&\leq \frac{c_R\eta^2}{\tau}  + c_{\textrm{HP}}\eta\log \tau \log^4 \frac{\eta}{\tau} +c_R\eta\log \tau\log^{5}\frac{c\eta}{\tau\kappa},
\end{align*}
where the inequality $c_{\textrm{HP}}+\frac{c}{\kappa}c_R \leq c_R$ holds for sufficiently large values of $c_R$ and a value of $\kappa<1$ that is larger than $c$ and sufficiently close to $1$ (say, $c_R=10c_{\textrm{HP}}$ and $\kappa = 9/10$). Now we focus on the last two terms of the inequality.  We upper bound $\log^{5}(c\eta/\tau\kappa)$ by $\log^4(\eta/\tau)\log(c\eta/\tau\kappa) =(\log^4(\eta/\tau))(\log(\eta/\tau) - \log(\kappa/c))$ and substitute to obtain

\begin{align*}
T(\eta,\tau) &\leq \frac{c_R\eta^2}{\tau}  + c_{\textrm{HP}}\eta\log \tau \log^4 \frac{\eta}{\tau} + \left(c_R\eta\log \tau\log^4\frac{\eta}{\tau}\right)\left(\log\frac{\eta}{\tau} - \log\frac{\kappa}{c}\right)\\
&\leq \frac{c_R\eta^2}{\tau}  + \left(\eta\log\tau\log^4\frac{\eta}{\tau}\right)\left(c_{\textrm{HP}} +c_R\log\frac{\eta}{\tau} - c_R\log\frac{\kappa}{c}\right)\\
&\leq \frac{c_R\eta^2}{\tau}  +c_R \left(\eta\log \tau\log^{5}\frac{\eta}{\tau}\right)\\ 
&= c_R \left(\eta^2/\tau  +\eta\log \tau\log^{5}\frac{\eta}{\tau}\right),
\end{align*}
as claimed.
Again, the $c_{\textrm{HP}} -c_R\log(\kappa/c)\leq 0$ inequality holds for sufficiently large values of $c_R$ that depend on $c_{\textrm{HP}}$, $\kappa$ and $c$.

\subsection{Space Bounds}
\label{sec:space}

We now show that the space bound holds. Recall that that we picked a parameter $\tau$ to bound the amount of space we use. Our algorithm uses more than $\tau$ space, but does not exceed $L \cdot \tau$ (for some large absolute constant $L>0$). 

First we count the amount of space needed in recursion. Our algorithm will stop the recursion whenever the problem instance fits into memory or $\tau$ becomes small (in the example we chose, when $\tau \leq 10$). Since the value of $\tau$ decreases by a constant factor at each level of recursion, we will never recurse for more than $\log_\kappa s=O(\log s)$ levels. Thus, the implicit memory consumption used in recursion does not exceed the space bounds.

Now we bound the size of the workspace needed by the algorithm at level $i$ of the recursion (with the main algorithm invocation being level~$0$) by $O(s\cdot \kappa^i)$. Indeed, this is the threshold of space we receive as input (recall that initially we set $\tau=s$ and that at each level we reduce this value by a factor of $\kappa$). This threshold value is the amount of space for the shortest-path computation algorithm invoked at the current level, as well as the limit on the number of vertices of $\pi$ that are stored explicitly before invoking procedure \textsc{FindAlternatingDiagional}. Once we have found the new alternating diagonal, the vertices of $\pi$ that were stored explicitly are used to generate the subproblems for the recursive calls. 

The space used for storing the intermediate points can be reused after the recursive executions are finished, so overall we conclude that at the $i$th level of recursion the algorithm never uses more than $O(s\cdot \kappa^i)$ space. Since we never have two simultaneously executing recursive calls at the same level, and $\kappa<1$ is a constant, the total amount of space used in the execution of the algorithm is bounded by 

\[
  O(s) + O(s\cdot \kappa) + O(s\cdot \kappa^2 ) + \ldots = O(s).
\]

\subsection{Considerations on the output}\label{sec:output}

For simplicity of the explanation, we assumed above that only edges of the triangulation needed to be reported. As mentioned in the introduction, our algorithm can be modified so that it reports the resulting faces (triangles) of the decomposition together with their adjacency relationship. 

For example, we could list all the triangles (say, as triples of vertex indices) and for each one we can give its adjacent triangles. Similarly, for each edge (identified by a pair of indices) we can also report the clockwise and counterclockwise neighbor at each endpoint, and so on. Recall that in our computation model the output cannot be modified, so all information about a
triangle should be output at the same time. For example, when we report the first triangle, we need to know the identities of its adjacent triangles, which we have not yet computed.

In order to accomplish this, we require that the space allowance $s$ be at least $\log n$. Recall that at each level of recursion the size of the problem decreases by a constant factor. In particular, if $s \geq \log n$, the algorithm does not run out of recursion space, and line~4 of Algorithm~\ref{algo_main} is never executed. 

That is, our algorithm partitions~$P$ into subpolygons~$P'$ until they fit into memory and triangulated using Chazelle's algorithm~\cite{c-tsplt-91}. Since the resulting triangulation of~$P'$ fits into memory, we can afford to report extra information. 

This extra information is explicitly available at the bottom level of each recursion (i.e., within a subpolygon $P'$), so we can report it together with the diagonals of the triangulation. The only information that we may not have available is for the diagonals that separate $P'$ from the rest of the polygon and for the triangles that use these edges. This information will appear in two subpolygons, and the neighboring information has to be coordinated between the two instances. 

For this purpose, we slightly alter the triangulation invariant associated with~$a_c$: subpolygon $P_1$ has been triangulated and all information has been reported \emph{except} for the diagonals (and the triangles that use those edges) between the two alternating diagonals. We explicitly store the pertinent adjacency information that has already been computed, and we report it only when at a later time the missing information becomes available.

This modification does not affect the running time or correctness of the algorithm. Thus, it suffices to show that space constraints are not exceeded either. 
At any given moment of the operation of the algorithm, we store a constant amount of additional data associated with each diagonal that delimits currently existing subproblems and that we already record. As shown in Section~\ref{sec:space} the diagonals themselves are stored explicitly and that storage fits into $O(s)$ storage. In particular, the additional information will not exceed this constraint either.

\section{Other applications}\label{sec:extensions}
Algorithm~\ref{algo_main} introduces a general approach of solving problems recursively by partitioning~$P$ into subpolygons, each of which has $O(s)$ vertices.  We focused on triangulating~$P$, so at the bottom of the recursion we used Chazelle's algorithm~\cite{c-tsplt-91} or Asano \etal's algorithm~\cite{abbkmrs-mcasp-11} depending on the available space. However, the same approach can be used for other structures: it suffices to replace the base cases of the recursion (lines 2 and 4 of Algorithm~\ref{algo_main}) with the appropriate algorithms. 

As an illustration of other possible applications, we describe the modifications needed for computing the shortest-path tree of a point inside a simple polygon and for partitioning a polygon into $\Theta(s)$ subpolygons, each with $\Theta(n/s)$ vertices. We believe other applications can be obtained using the same strategy. 

\subsection{Shortest-Path Tree}
\label{sec:spt}

Given a simple polygon $P$ and a point $p\in P$ (which need not be a vertex of $P$), the \emph{shortest-path tree} of $p$ (denoted by $\SPT(p)=\SPT(p,P)$, see Figure~\ref{fig:ShortestPathtree}, left)  is the tree formed by the union of all geodesics from $p$ to vertices of $P$. ElGindy~\cite{e-hdpa-85} and later Guibas \etal~\cite{ghlst-ltavsppitsp-87} showed how to compute the shortest-path tree in linear time using $O(n)$ space. 
In order to use the framework of Algorithm~\ref{algo_main}, we also need an algorithm that computes $\SPT(p)$ using a constant number of variables.

\begin{figure}
  \centering
  \includegraphics[width=\textwidth]{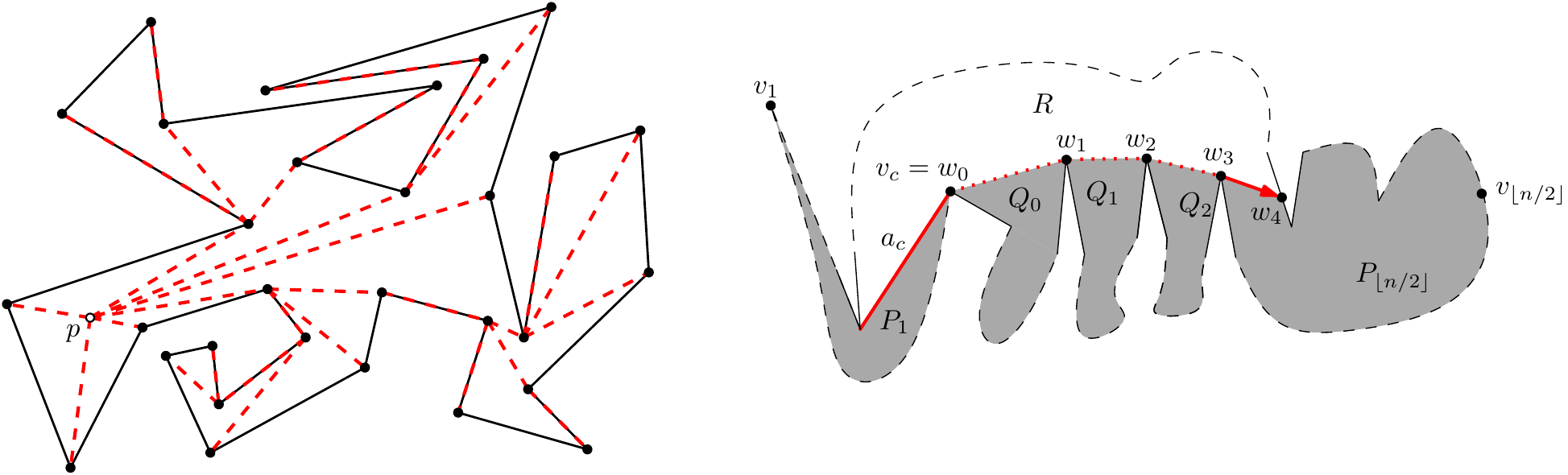}
  \caption{(left) The shortest-path tree for a point $p$, depicted with red dashed segments. (right) In the framework of Algorithm~\ref{algo_main}, we use a different method to find an alternating diagonal after walking for $\tau$ steps ($w_4$ is now the first proper intersection between the ray from $w_2$ towards $w_3$ and the boundary of $P$). Further note that paths from all vertices within a subpolygon corresponding to a recursive subproblem pass through a common vertex (the paths from $Q_i$ pass through $w_i$, for example). Thus, in each subpolygon we can forget about $v_1$ and generate up to $\tau+2$ independent subproblems.}
  \label{fig:ShortestPathtree}
\end{figure}

\begin{lemma}\label{lem_o1wspace}
Let $P$ be a simple polygon with $n$ vertices and let $p$ be any point of $P$ (vertex, boundary, or interior). We can compute $\SPT(p)$ in $O(n^2\log n)$  expected time using $O(1)$ variables.
\end{lemma}

\begin{proof}
We first show a randomized procedure that, given a simple polygon, a source~$q$, and a target $t$ computes the first link in the shortest path~$\sigma$ from~$q$ to~$t$ in expected $O(n\log n)$ time using $O(1)$~space. Our algorithm executes this procedure $n$~times, setting $q$ to be each vertex of $P$ in turn, and reporting the first edge towards~$p$. The union of these segments is $\SPT(p)$.

Thus, it suffices to show how to compute one edge of $\sigma$ efficiently. The constant-workspace shortest-path algorithm of Asano~\etal~\cite{amrw-cwagp-10} computes the entire shortest path $\sigma$ from $q$ to $t$ in $O(n^2)$ time, but computing a single segment of $\sigma$ may need $\Omega(n^2)$ time. Below we slightly modify their approach to ensure that we do not spend too much time in one step.

Note that we will only need this procedure when $q$ is a vertex of $P$; thus, for simplicity of presentation, we will assume so in the remainder of the proof.  The more general case can be handled with only minor modifications. We begin by assuming that the first link of $\sigma$ lies in the interior of a given cone~$C$ with apex~$q$ (see Figure~\ref{fig:step}); initially $C$ is delimited by the directions of the edges incident to~$q$.  Let $R_C$ be the set of all reflex vertices of $P$ lying in the interior of~$C$ (we include $t$ in $R_C$ as well if it lies in the interior of~$C$).

\begin{figure}
  \centering
  \includegraphics[width=0.9\textwidth]{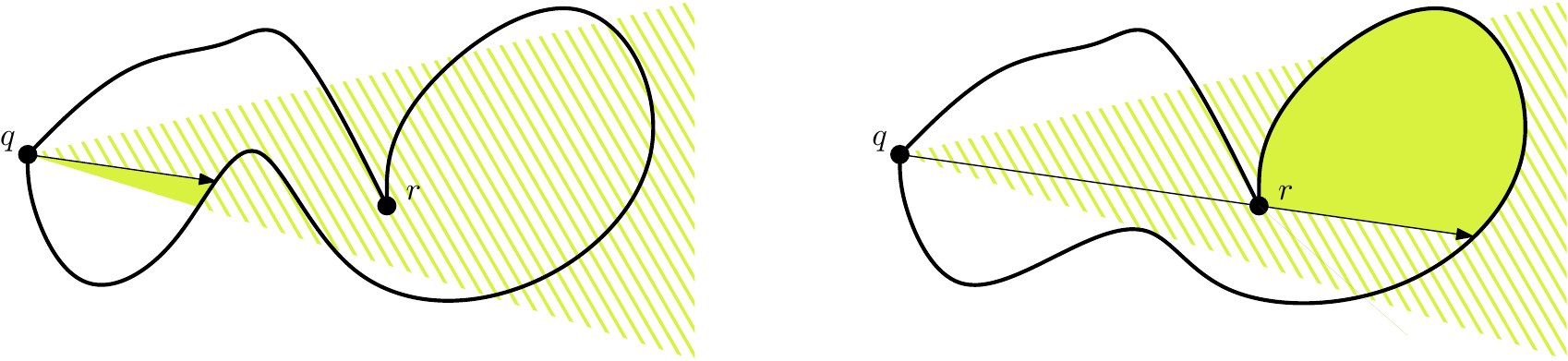}
  \caption{Given the cone $C$ (dashed green), we shoot a ray towards a reflex vertex~$r$ in $C$. If $r$ is not visible (left image), the cone $C$ is split into two by the ray (solid and dashed regions). The region that does not contain $t$ can be discarded (the resulting cone will have the ray as one of the new boundary edges).  If $r$ is visible (right image) the ray splits $P$ into three components, one of which is hidden (i.e., the only visible portion of the hidden component from $q$ is in the ray itself). If $t$ lies in the hidden region (solid in the figure) we can report $r$ as the first vertex visited on the path from $q$ to $t$. Otherwise, we can shrink $C$ in a similar way as if $r$ were not visible. }  
  \label{fig:step}
\end{figure}
By our assumption, the first link of $\sigma$ must go towards a point of $R_C$, even if $\sigma$ leaves the cone later. So, if $R_C$ contains no reflex vertex, then $\sigma=qt$, the algorithm returns~$t$, and we are done.  Otherwise, we pick a reflex vertex $r\in R_C$ uniformly at random and shoot a ray $\vec{qr}$ from~$q$ towards~$r$.  Let $q'$ be its first proper intersection with the boundary of~$P$.

The cone $C$ is split into two cones by the ray~$\vec{qr}$, and the segment $qq'$ splits $P$ into two or three components, depending on whether or not $q$~sees~$r$. We compute the component $P'$ that contains~$t$. The following cases may occur (refer to Figure~\ref{fig:step}).
\begin{description}
\item[$q$ does not see $r$:] Then, $\sigma$ cannot go directly from $q$ to $r$. Moreover, $\sigma$ cannot enter $P\setminus P'$, so its first link must emanate from $q$ into $P'$.  This determines the side of $\vec{qr}$ it must lie on.
Thus, we may shrink $C$ to a smaller cone and continue.

\item[$q$ can see $r$:] In this case $qq'$ splits $P$ into three components, two of which contain $q$ on their boundary, and the third one (the \emph{hidden} component) that does not. If $t$ does not lie in the hidden component, then we can shrink $C$ to a smaller cone with the same analysis as above and continue.  If $t$ lies in the hidden component,
$qr \subset \sigma$, the algorithm returns~$r$, and we are done.
\end{description}

See Asano \etal~\cite{amrw-cwagp-10} for proof of correctness and how to handle degenerate cases. Overall, in one iteration we either find the desired vertex and terminate, or reduce the size of $R_C$. Each step can be done in linear time (we perform one ray-shooting operation and one point-location operation in constant workspace; each can be done in $O(n)$ time by brute force).
Our algorithm executes a procedure analogous to a randomized binary search on the set of directions to the vertices in $R_C$, bisecting it randomly and recursing on one of the ``halves,'' and terminating (at the latest) when this set is a singleton.  Therefore the expected number of iterations is logarithmic and the total expected work required to find the first link of $\sigma$ is $O(n \log n)$.
\end{proof}

Note that we can make the above algorithm deterministic by using selection instead of picking a vertex of $R_C$ at random. This comes at a slight increase in the running time as a function of $s$ (see the detailed trade-off description and analysis in~\cite{bkls-cvpufv-13}).

Since we now have algorithms for $\Theta(1)$ and for $\Theta(n)$ words of working memory, we can use our general strategy to obtain a trade-off for the entire range of the space parameter~$s$.
 
\begin{theorem}
Let $P$ be a simple polygon with $n$ vertices and let $p$ be any point of $P$ (vertex, boundary, or interior). For any $s \leq n$ we can compute the shortest-path tree of $p$, $\SPT(p)$, in $O(n^2\log n/s+n\log s\log^{5} (n/s)$ expected time using $O(s)$ variables.
\end{theorem}
\begin{proof}
In order to use the framework of Algorithm~\ref{algo_main}, we first ensure that $p$~is a vertex of the polygon.  If $p$ is already a vertex of $P$, we rename the vertices so that $p = v_1$. If $p$ lies in the interior of an edge of $P$, we look for a vertex~$q$ visible from $p$. The segment $pq$ splits $P$ into two subpolygons, and we run Algorithm~\ref{algo_main} on each subpolygon separately, renaming the vertices so that $p=v_1$. Although $pq$ appears in both shortest-path trees, we make sure it is only reported in one of the two subproblems. Finding a visible vertex $q$ can be done in linear time using a constant number of variables as explained in the proof of Lemma~\ref{lem:farAlternation}.  Finally, if $p$ lies in the interior of~$P$, we find two visible vertices $q,q'$, again using the approach of Lemma~\ref{lem:farAlternation}. The segments $pq$ and $pq'$ split $P$ into two subpolygons both of which have $p$ as a vertex. As in the boundary case, treat the two subpolygons independently to obtain the overall tree, with $pq$ and $pq'$ reported once.  In all cases, we introduce a constant number of modifications to the polygon, so they can be stored explicitly.

Now that $p=v_1$ is a vertex of $P$, we use the overall approach of Algorithm~\ref{algo_main}: compute the shortest path from $v_1$ to $\vmed$. Alternating diagonals found along the path are again used to generate subproblems,  which are solved recursively until we run out of space. At the bottom of the recursion we use a linear-time algorithm for computing $\SPT(p)$ (such as those of ElGindy~\cite{e-hdpa-85} or Guibas \etal~\cite{ghlst-ltavsppitsp-87}) or Lemma~\ref{lem_o1wspace}, depending on whether or not the remaining polygon fits in memory.

We must slightly modify the way the algorithm finds an alternating diagonal when it has performed $\tau$ steps: we need a diagonal that ensures that both subproblems are independent (in contrast to the triangulation problem, where any diagonal suffices). Instead, we simply extend the edge $w_{\tau-1}w_{\tau}$ until it meets the boundary of~$P$.  We declare this intersection point a virtual vertex $v$ (if it is not already a vertex) and use the segment $w_{\tau} v$ to split the polygon (see Figure~\ref{fig:ShortestPathtree}, right). Since $w_{\tau-1} w_{\tau}$ is an edge of the geodesic path, $v_1$ and $\vmed$ are on opposite sides of $w_{\tau} v$, thus $w_{\tau} v$ is an alternating diagonal.\footnote{Note that it is not properly a diagonal since there might not be a vertex at $v$, but by adding the virtual vertex we can treat it as one.  We should remember to ignore it when outputting shortest-path tree edges.}

In each subpolygon we want to compute the shortest-path tree to $v_1$ which may lie outside the current subpolygon.  Instead, we will show that in each subpolygon $P'$ there exists a vertex $w$ such that $\SPT(v_1,P)\cap P'=\SPT(w,P')$. 

Indeed, the boundary of $P'$ consists of a contiguous portion of the boundary of $P$ and up to $s$ diagonals.  Recall that in all cases these diagonals belong to the shortest path from $v_1$ to a boundary point of~$P$.  These diagonals form a contiguous portion $\pi'$ of a shortest path to~$v_1$. Let $w$ be the vertex of $\pi'$ closest to $v_1$. Let $q$ be any point in $P'$. Since two shortest paths to $v_1$ cannot cross, we conclude that the shortest path from $q$ to $v_1$ cannot properly intersect~$\pi'$. Thus, after it intersects with it, it must follow the same path towards~$v_1$. In particular, it must also pass through $w$, which implies $\SPT(v_1,P)\cap P'=\SPT(w,P')$, as claimed.

That is, when processing a small subpolygon, we can forget about $v_1$ and compute the shortest-path tree to $w$, giving the same structure as in the original problem. 
By doing so, we ensure that the recursively split polygons have the same structure as in Algorithm~\ref{algo_main}: a chain of contiguous input vertices and a (small) number of cut vertices stored in memory. 

The analysis of space use is identical to that of Algorithm~\ref{algo_main}. We now turn to the running time bound.  We claim that, for a suitably chosen constant $c$, it obeys the recurrence

\[
  T(\eta,\tau) = c \left(\eta^2\log \frac{\eta}{\tau} + \eta \log \tau \log^4 \frac{\eta}{\tau} + \sum_i T(\eta_i,\tau \kappa)\right),
\]
which differs from the recurrence in Section~\ref{sec:time} in that the constant-space running time algorithm of Lemma~\ref{lem_o1wspace} is slower than its counterpart in Algorithm~\ref{algo_main} by a $\log(\eta/\tau)$ factor.  By an entirely analogous analysis, the recurrence solves to

\[
  T(\eta,\tau) = O\!\left(\frac{\eta^2 \log \eta}{\tau}+ \eta\log \tau \log^{5} \frac{\eta}{\tau}\right),
\]
concluding the proof of the theorem.
\end{proof}

\subsection{Partitioning $P$ into subpolygons of the same size}
\label{sec:subpolygons}
Asano \etal~\cite{abbkmrs-mcasp-11} observed that one can use a triangulation algorithm to partition a polygon into pieces of any desired size. Specifically, they showed that in any simple polygon there always exist $\Theta(s)$ non-crossing diagonals that split it into subpolygons with $\Theta(n/s)$ vertices each. 

The existence was proven for any value of $s$ and the proof is constructive. However, since no time-space trade-off for triangulating polygons was known at that time, their algorithm would always run in quadratic time regardless of the size of the workspace (see Theorem~5.2 of~\cite{abbkmrs-mcasp-11}). We can now extend this result to obtain a proper time-space trade-off.

\begin{theorem}\label{theo_partiti}
Let $P$ be a simple polygon with $n$ vertices. For any $s\leq n$, we can partition~$P$ with $\Theta(s)$ non-crossing diagonals, so that each resulting subpolygon contains $\Theta(n/s)$ vertices. This partition can be computed in $O(n^2/s+n\log s\log^{6} (n/s))$ expected time using $O(s)$ variables. 
\end{theorem}
\begin{proof}
Just as for the shortest-path tree computation, one can modify Algorithm~\ref{algo_main} to partition a polygon into pieces for the entire range of available working space memory values. Alternatively, we can also do it by combining Theorem 5.2 of~\cite{abbkmrs-mcasp-11} with our triangulation algorithm (Theorem~\ref{main_theo}). Below we sketch a proof of the latter approach, for completeness; we omit some of the bookkeeping details; refer to \cite{abbkmrs-mcasp-11} for the specifics.

The algorithm makes several scans of the input. At each step we keep a partition of~$P$ into subpolygons $\mathcal{P}=\{P_1, \ldots, P_k\}$; initially $k=1$ and $P_1=P$. Let $t := \lceil n/s \rceil$, our aim is to iteratively cut the polygons of $\mathcal{P}$ into smaller pieces until they have between $t$ and $t/6$ vertices each. 

In each round, we scan each polygon~$P_i$. The ones with more than $t$ vertices are triangulated. For each edge of the triangulation, we check if it would create a balanced cut (i.e., a diagonal of a polygon of $n$ vertices makes a \emph{balanced cut} if neither component has fewer than $n/6$ vertices). 
It is known that such a cut always exists in any triangulation.  Once found, we use it to split the current polygon into two. After the $i$th round we have split~$P$ into subpolygons such that each either has the desired size or has at most $(5/6)^i n$ vertices. In each round we triangulate each subpolygon at most once. Moreover, each subpolygon is triangulated independently, so we can bound the running time of the $i$th round by

\begin{align*}
  \sum _j \left(\frac{n_j^2}{s}+n_j\log n\log^{5}
  \frac{n_j}{s}
  \right)
& = 
  \sum _j \frac{n_j^2}{s} +  \log n\log^{5}
  \frac{n}{s}
  \sum_{j} n_j\\
&\leq
  \left(\frac{5}{6}\right)^{i-1} n \sum _j\frac{n_j}{s}+ n \log n\log^{5}
  \frac{n}{s} \\
&\leq
  \left(\frac{5}{6}\right)^{i-1} \cdot 3 (n^2/s) +n\log n\log^{5} \frac{n}{s},
\end{align*}
where we have used the fact that in each round the subpolygon sizes add up to at most~$3n$.
Summing over all passes and observing that the number of passes is at worst logarithmic in $n/s$, we conclude that the running time of this algorithm is
$O(n^2/s +n\log n\log^{6} (n/s)).$

Regarding space, each triangulation algorithm we invoke uses $O(s)$ space. Since each execution is independent we can reuse the space each time. In addition to that we need to explicitly store the list $\mathcal{P}$. Since all polygons of $\mathcal{P}$ have at least $t/6 \in \Theta(n/s)$ vertices, we  never maintain more than $O(s)$ such subpolygons. Thus, the space bounds are also preserved. 
\end{proof}

We note that the shortest-path algorithm of Har-Peled~\cite{Har-Peled15} also partitions $P$ into $O(s)$ pieces (of size $O(n/s)$ each) as part of his preprocessing. However, this is done by introducing Steiner points. Our approach can report the partition implicitly (by giving the indices of the diagonals) and avoids the need for Steiner points.

\section{Acknowledgments}
The authors would like to thank Jean-Fran\c{c}ois Baffier, Man-Kwun Chiu, and Takeshi Tokuyama for valuable discussions that preceded the creation of this paper. 
Moreover, we would like to thank Wolfgang Mulzer for pointing out a critical flaw in a preliminary version of the paper, as well as for his help in correcting it.

\bibliography{triang}

\end{document}